\pdfoutput=1
\documentclass[journal]{IEEEtran}
\usepackage{palatino}
\usepackage{tikz}
\usetikzlibrary{shapes.geometric, arrows}
\usepackage{graphicx}
\usepackage{float}
\usepackage{color}
\usepackage{amsfonts}
\usepackage{subcaption}
\usepackage{caption}
\usepackage{cite}

\ifCLASSINFOpdf
 \else
  \fi
\usepackage{amsthm,bm}
\newtheorem{theorem}{Theorem}

\theoremstyle{definition}

\theoremstyle{remark}

\usepackage{multirow}
\usepackage{amsmath}
\begin{document}
\title{Quantifying Hosting Capacity for Rooftop PV System in LV Distribution Grids\vspace{-3.5mm}}
\author{Jingyi~Yuan,~\IEEEmembership{Student~Member,~IEEE,}
Yang~Weng,~\IEEEmembership{Member,~IEEE,}
and~Chin-Woo~Tan,~\IEEEmembership{Member,~IEEE}\vspace{-12mm}
}
\maketitle
\vspace{-6mm}
\begin{abstract}
Power systems face increasing challenges on reliable operations due to the widespread distributed generators (DGs), e.g., rooftop PV system in the low voltage (LV) distribution grids.
Characterizing the hosting capacity (HC) is vital for assessing the total amount of distributed generations that a grid can hold before upgrading. 
For analyzing HC, some methods conduct extensive simulations, lacking theoretical guarantees and can time-consuming. Therefore, there are also methods employing optimization over all necessary operation constraints. But, the complexity and inherent non-convexity lead to non-optimal solutions.
To solve these problems, 
this paper provides a constructive model for HC determination.
Based on geometrically obtained globally optimal HC, 
we construct HC solutions sequentially according to realistic constraints, so that we can obtain optimal solution even with non-convex model. For practical adaption, we also consider 
three-phase unbalance condition and  parallel computation to speed-up. We validated our method by extensive numerical results on IEEE standard systems, such as $8$-bus and $123$-bus radial distribution grids, where the global optimums are obtained and performance illustrations are demonstrated.
\vspace{-4mm}
\end{abstract}
\IEEEpeerreviewmaketitle
\vspace{-2mm}
\section{Introduction}
\vspace{-1mm}
\IEEEPARstart{I}{n} recent years, an increasing number of renewable energy-based distributed 
generators (DGs) appears in power distribution grids. Such widespread use of DGs brings great benefits, e.g., voltage profile support, loss reduction, and lower capital cost. But, such change also raises operational challenges, like power quality problems and system control issues \cite{E2012DGimpact}. 
For example, high-level DG penetration can cause undesirable voltage flicker or equipment aging due to the excessive voltage regulating behaviors \cite{p2000determining, G2011impact}. To avoid these problems, evaluating the hosting capacity of a system is vital to understand the impact of DG penetration level and conduct upgrades for reliable system operation. 

Hosting capacity (HC) is defined as the maximum amount of the power generation that a system can host without violating any operating standards \cite{JSmith2012EPRI3}. 
With such a definition, Electric Power Research Institute (EPRI) proposes simulation-based methods, modeling each feeder and examine all the power quality and reliability issues with screening tools to determine HC at different locations 
\cite{M2016streamlined,l2016utility,n2016determination}. For example, we can keep on increasing the PV (photovoltaics) penetration level for power flow analysis until constraint violation, where the last feasible PV generation is regarded as the HC. 
Such idea has already been deployed in commercial software. For example, our collaboration partner, CYME, focuses on the hosting capacity of one node per analysis \cite{CYME}. 
Fig. $1$ shows the three steps how CYME obtain the per-bus hosting capacity. 
However, no matter if one is focusing on one bus or for all buses, 
it is numerically impossible to exhaust all possible scenarios with simulations. 
\begin{figure}
\centering
\includegraphics[width=3.5in]{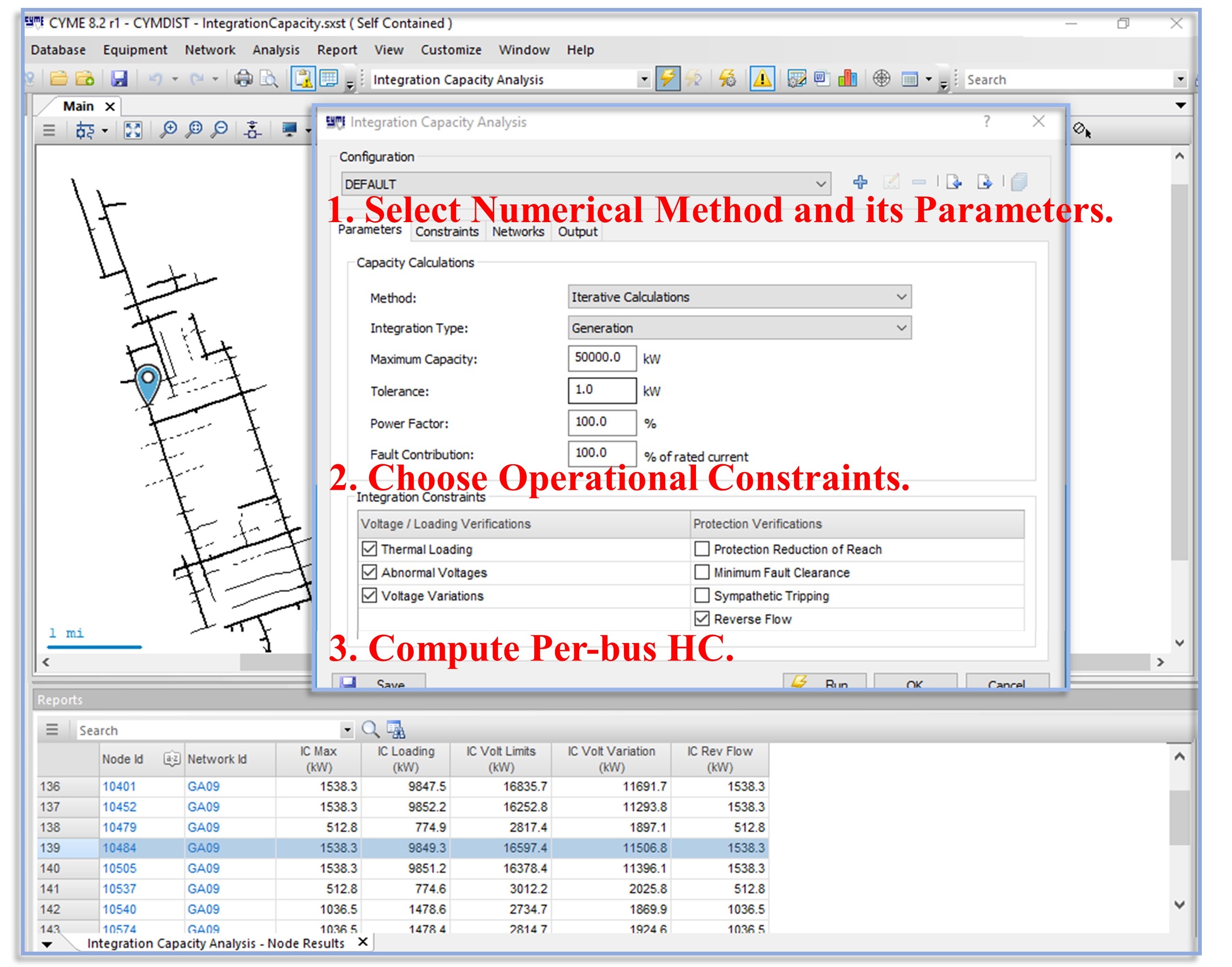}
\label{figure_4_wire_model}
\vspace{-2mm}
\caption{CYME for numerical hosting capacity analysis.} 
\vspace{-4mm}
\end{figure}

Therefore, there are also mathematical model-based methods. 
For example, \cite{S2002maximising} formulates HC problem as an AC optimal power flow 
(ACOPF) problem, extended by  \cite{L2010distribution,f2015assessing} for a multi-period ACOPF. As solving ACOPF with many constraints are hard, \cite{F2017evaluation,A2015understanding6} propose to conduct sensitivity analysis as an approximation to characterize key factors, e.g., voltages, the network topology, load size, voltage limits, and generation locations \cite{S2002maximising,k2014locational,c2005voltage}. There also recent study with a probabilistic model on key factors \cite{abad2018probabilistic}.

Based on these factors, additional control strategies can improve HC via $1$) on-load tap-changers (OLTC) technology \cite{d2016increasing},
$2$) reactive power control \cite{n2011increasing,divshali2017improving,j2015improving}, 
$3$) active network management \cite{L2010distribution}, 
$4$) static and dynamic network re-configurations \cite{f2015assessing}, 
and $5$) adding soft open points from 
power-electronic devices \cite{l2016assessing}, and $6$) incorporating real-time 
information \cite{S2002maximising}. 
Besides, one can also model economic aspects to do a cost-benefit analysis for maximizing HC \cite{s2014optimizing}. 

However, no matter how to approximate, e.g., selecting key features, the aforementioned studies can not locate the globally optimal solution. 
The objective in these methods is also inflexible to reduce optimizing variable numbers. 
Therefore, this paper focuses on obtaining exact solutions constructively with a flexible objective for unevenly deployed DGs. 

The first contribution is changing the objective in ACOPF to a flexible form and sequentially adjusting solution with constraints to preserve convexity as much as possible.
For example, we transform into rectangular coordinates to convexify certain constraints. 
When convexification is impossible, our second contribution lies in connecting geometric understanding for reaching a global optimum even when the constraint is still non-convex. A detailed proof is provided. 
Our third contribution is for providing practical adjustments, where we extend
our static model with sequence components for 
three-phase unbalance situation. We also provide a method for paralleling computation  for large systems. 

Extensive simulations verify the performance of the proposed method on IEEE distribution networks, e.g., $8$-bus and $123$-bus. The results show that our problem formulation and proposed theory can find exact HC in diversified scenarios. It also validates the integrated methods for multi-phase unbalanced situation and the distributed-optimization for computational reduction.

The rest of the paper is organized as follows: Section \uppercase\expandafter{\romannumeral2} shows the problem formulation. Section \uppercase\expandafter{\romannumeral3} proposes a 
generalized theory of maximum HC solution based on 
geometric understanding and the proof. Section \uppercase\expandafter{\romannumeral4} shows the algorithm adjustments for three-phase unbalanced system and parallel computation. Section \uppercase\expandafter{\romannumeral5} tests the theory in various systems and Section \uppercase\expandafter{\romannumeral6} concludes the paper.
\vspace{-4mm}
\section{Mathematical Remodeling}
The power flow equations are
\vspace{-1mm}
\begin{subequations}  \label{eq:1}
\begin{align}
 \vspace{-2mm}
P_{i}^{inj}=\sum_{k=1}^n {|V_{i}|}{|V_{k}|}(G_{ik}\cos\theta_{ik}+B_{ik}\sin\theta_{ik}), \label{eq:1a}
\\ \vspace{-2mm}
Q_{i}^{inj}=\sum_{k=1}^n {|V_{i}|}{|V_{k}|}(G_{ik}\cos\theta_{ik}-B_{ik}\sin\theta_{ik}),\label{eq:1b}
\vspace{-1mm}
\end{align} 
\end{subequations}
where $n$ is the bus number of the electrical system, $P_{i}^{inj}$ is the active power injection at bus
$i$, $V_{i}$ is the voltage magnitude at bus $i$, and $G_{ik}+jB_{ik}$ is the element in the bus admittance matrix $Y_{\textrm{Bus}}$ $\in \mathbb{C^{\mathit{n}\times \mathit{n}}}$. Specifically, $G_{ik}$ represents the line admittance between bus $i$ and bus $k$.

For flexible objective, we include a free parameter $\lambda$ in front of each potential generator. 
 \vspace{-1mm}
$$
HC = \sum_{i=1}^n \lambda_i \cdot P_{i}^{inj} = \sum_{i=1}^n \lambda_i \cdot (P_{i}^{gen}-P_{L}^{load}),  \vspace{-1mm}
$$ 
Here, $\lambda$ can be a generating scaling factor, indicating how fast a linearly increas with ``base'' generation can be. 
For example, $\lambda_1=\lambda_2=\cdots=\lambda_n$ are widely assumed in literature for simplification, where we are looking into an ideal and even growth for situational awareness. 
However, practical setup won't be homogeneous in general, 
therefore, our model provide such a flexibility. 
\begin{itemize}
\item when $\lambda_{i}$ is a binary parameter, $\lambda_{i}=1$ indicates bus $i$ has solor generation and $\lambda_{i}=0$ vise verse;
\item when $\lambda_{i}$ is a non-negative parameter, e.g., $\lambda_{i} \in \mathbb{R+}$, the weight parameter represents the importance of a bus, e.g., the relative price of installing solar panels.
\end{itemize}
In this paper, we will first let $\lambda_{i}=1, \forall i$ for simplification and focus on the variable $P_{i}$ for geometric insights, e.g., 
\vspace{-1mm}
$$ 
\sum_{i=1}^n \lambda_{i} \cdot P_{i}^{inj} = \sum_{i=1}^n  P_{i}^{inj}.
$$
Then, we will remove such an assumption for generality. 

Operational constraints must be satisfied while calculating hosting capacity.
\begin{subequations} 
\label{eq:2}
\begin{align}
\max &\sum_{i=1}^n  \lambda_i\cdot P_{i}^{inj} \label{eq:2a}\\
\rm{s.t.} \
 & \text{Equ.} \ \eqref{eq:1}, 
 \label{eq:2b}\\
&V_{min}\leq |V_{i}| \leq V_{max},\label{eq:2c}\\
&\theta_{min}\leq \theta_{ik} \leq \theta_{max},\label{eq:2d}\\
&I_{Lmin}\leq I_{ik}\leq I_{Lmax},\label{eq:2e}\\
&pf_i\geq \eta,\label{eq:2f}
 \vspace{-2mm}
\end{align}
\end{subequations}
where 
\eqref{eq:2b} is the power-flow constraint. \eqref{eq:2c} and \eqref{eq:2d} state the constraints over voltage magnitudes 
and the voltage angle difference. 
\eqref{eq:2e} represents the thermal limit. 
\eqref{eq:2f} is the power factor requirement of generator bus, where 
$pf_i=\frac{(P^{inj}_i)^2}{(P^{inj_i})^2+(Q^{inj_i})^2}$.  
  \vspace{-1.5mm}
\section{Finding the Global Optimum of Hosting Capacity based on Geometrical Understanding }
\subsection{Geometrical Intuition for Obtaining Hosting Capacity}
It is challenging to obtain global optimum in \ref{eq:2}, due to the nonlinear correlations of the optimization variables in power flow equation, let alone the non-convex constraints.
However, for a small system such as a $3$-bus system, hosting capacity can give geometrical results and intuition, shown in Fig. $2$. 

For example, we can let bus $0$ be the reference bus with voltage $1\angle0^\circ$. Let bus $1$, $2$ be PV-type buses and let the branch resistance be $R=1$. For simplicity, let the voltage angle differences and line reactances be zeros.
Based on the setup, we show below two geometric intuitions to find the global optimum. 

\vspace{-1mm}
     \begin{figure}[H]
\vspace{-2mm}     
         \centering
        \includegraphics[width=3in]{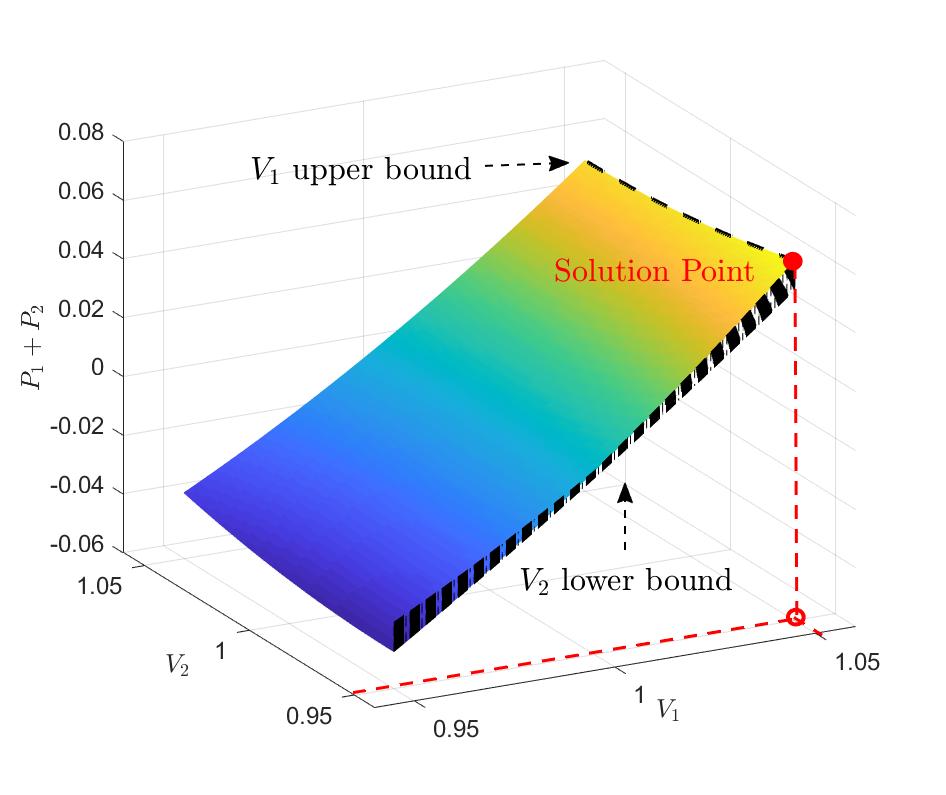}
          \vspace{-2mm}
         \caption{A $3$D plot of $P_{1}+P_{2}$, $V_{1}$, and $V_{2}$ that illustrates the P-V (total power-individual voltage) relevancy for the $3$-bus network. The value range of both $V_{1}$ and $V_{2}$ respect to the total power makes it a surface. The solid red point denotes the operation point with maximum generation where $V_{1}$ is at upper bound and $V_{2}$ is at lower bound. }
         \label{fig:3d}
     \end{figure}

For power system analysis, nodal P-V (power-voltage) curve is commonly used which is a 2-D plot.
To find the HC due to variable correlations, we extend to a 3-D P-V plot in Fig. 2a.
In the $3$-bus visualization, the only voltage constraint ($ 0.95\leq|V|\leq 1.05$) leads to the feasible solution region.
If $V_2$ is fixed, we observe that the curve monotonically increases as $V_1$ increases, which leads to finding 
the possible solution at the top edge. Then we find the maximum generation capacity at one vertex where $V_2$ is at lower bound. As marked by a red point, the solution of HC is obtained when  $V_1=1.05 \ p.u.$ is the highest and $V_2=0.95 \ p.u.$ is the lowest.
Such a solution pattern indicates hosting capacity be in direct proportion to the
difference of the voltage magnitudes over the neighbor buses. 
          \vspace{-2mm}
\begin{figure}[H]
          \vspace{-2mm}
         \centering
         \includegraphics[width=3in]{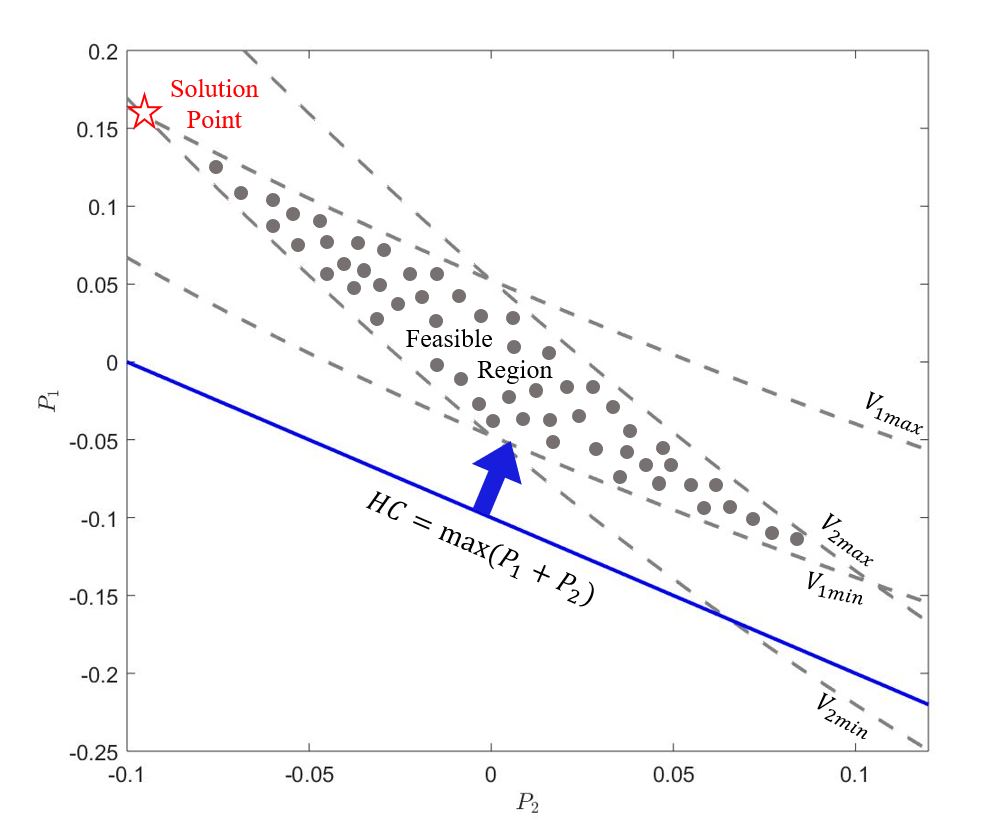}
          \vspace{-1mm}
         \caption{The plot of $P_{1} \ vs. \  P_{2}$. The voltage constraints bound a feasible region under power flow equations. We use a straight line: $HC= \max \ (P_{1}+P_{2})$ to cut the maximum value. The same solution point is found considering pairwise power correlation. It denotes that bus $1$ generates $0.1575 \ p.u.$ power while bus $2$ consumes $0.095 \ 
         p.u.$  power.}
         \label{fig:2d}
          \vspace{-2mm}
     \end{figure}
      \vspace{-2mm}
In addition to voltage correlation, it is important to validate voltage solution more directly in the pairwise power domain to understand how voltage constraints impact the optimum solution. For this purpose, Fig. 2b is created by firstly plotting all possible $(P_1,P_2)$ pairs as the gray dots according to power flow equations. Secondly, we add the voltage upper bounds and lower bounds of bus $2$ and bus $3$ to eliminate points outside the allowed region. Then, we push the line of $P_1+P_2$ towards bigger values until we reach the global optimum. 
The red star in Fig. 2b confirms that the same optimal HC is reached when the neighbor buses' voltages are 
with the pattern of one high and one low. The positive value of $P_1$ indicates generation while the negative value of $P_2$ means it absorbs power from bus $1$ and hence the total is the maximum. Such hosting capacity take neighbour-bus correlation into consideration that allows all possible patterns for different buses. 

We assign relative numbers to
represent the bus order, e.g., bus 1 is odd and bus 2 is even. 
To observe mathematically, the opposite
but extreme values for neighbor buses attain individual maximization for the quadratic functional component in
power flow equations. The conclusion is feasible in both resistive and general inductive network.
Based on the geometric summary, we can infer the general theory to the solution and the condition to achieve it next.  

\vspace{-3mm}
\subsection{Mathematical Theorem for Global Optimum Solution}
\subsubsection{HC due to Voltage Magnitude Constraints}
Based on the geometric understanding on voltage magnitudes bounds, we 
can generalize the results from $3$-bus to $n$-bus radial network 
high-low-voltage pattern.
\begin{theorem}\label{th1}
Let a radial distribution network be with resistive or inductive line impedance and let the voltage magnitude constraint be the only constraint on HC. The maximum power generation in ($3$), $HC^*= HC_{max}$, is obtained when $|V_{odd}|=V_{max}$ and $|V_{even}|=V_{min}$. 
\end{theorem}
For the proof, it is sufficient to show that the proposed hosting value $HC^{*}$ is larger than arbitrary summation of power generations, e.g., 
\vspace{-2mm}
\begin{align}
HC^{*}=\sum_{i=1}^n P_{i,\rm{solution}} \geq \sum_{i=1}^n {P_{i}}, \label{proof}
\vspace{-2mm}
\end{align}
\begin{proof}
See Appendix B.
\end{proof}

\subsubsection{HC due to Voltage Angle Constraints}
We add constraints on the angle difference, leading to complex voltage phase.
\begin{theorem}
Let $\theta_{min}=0\leq \theta_{i} \leq \theta_{max}$.
The HC is given by 
\begin{align*}
&\theta_{ik}=\min[\pi,\theta_{max}], |V|=\\
&\begin{cases}
V_{max} \ \ \textrm{for all buses},\ \textrm{if} \ \theta_{max}>\arccos(\frac{V_{max}+V_{min}}{2V_{max}}); \\
\begin{cases}
V_{max} \ \ \textrm{for bus} \ 2n-1,\ 
\\
V_{min}\ \ \textrm{for bus} \ 2n,\ \textbf{if}\ 0 <\theta_{max}\leq \arccos(\frac{V_{max}+V_{min}}{2V_{max}}),\\
\end{cases}
\end{cases}
\end{align*}
where $\arccos(\frac{V_{max}+V_{min}}{2V_{max}})$ is the critical value of 
the voltage angle boundary that changes the solution pattern. 
\end{theorem}
\begin{proof}
When the angle constraint is added, we can still start with $\sum_{i=1}^NP_{i}$ and obtain
\begin{align}
\label{rec}
\sum_{i=1}^NP_{i}=& \sum_{i=1}^{N}\sum_{k=1}^{N}(-G_{ik})\cdot[(V_{i,\textrm{re}}-V_{k,\textrm{re}})^2\\
&+(V_{i,\textrm{im}}-V_{k,\textrm{im}})^2]. \nonumber
\end{align}
The most valuable part is this quadratic form, although, sinusoidal components about angle is introduced. Detailed proof of the derivation of (\ref{rec}) can be found in Appendix B. Subsequently, (\ref{rec}) can be written in the polar coordinates to quantify the impacts of voltage magnitude and phase angle separately. 
\begin{align}
\begin{split}
\sum_{i=1}^NP_{i} = \sum_{i=1}^{N}\sum_{k=1}^{N}&(-G_{ik})\cdot {(V_{i}\cos\theta_{i}-V_{k}\cos\theta_{k})}^2\\
+&(-G_{ik})\cdot{(V_{i}\sin\theta_{i}-V_{k}\sin\theta_{k})}^2.
\end{split}
\end{align}
As the multiplication is the key component in the quadratic optimization, we can simplify the objective further by having 
${|V_{i}|}=a, ~{|V_{k}|}=b, ~\cos\theta_{i}=a_{1},~\sin\theta_{i}=a_{2}, ~\cos\theta_{k}=b_{1},~\sin\theta_{i}=b_{2}$.
Suppose bus $i$ is the even bus and bus $k$ is the odd bus.
Such notation reshape the problem model into
\begin{align}
\max & \sum_{i=1}^n(-G_{ik})\cdot\left[{(aa_{1}-bb_{1})}^2+{(aa_{2}-bb_{2})}^2\right] \label{objtran}\\
\rm{s.t.}  \  &V_{max}\leq a \leq V_{min}, \  \ V_{min}\leq b \leq V_{max},\nonumber\\
& \ a_{1}^2+a_{2}^2=1, \ \  b_{1}^2+b_{2}^2=1.\label{equcon}
\vspace{-2mm}
\end{align}
Substitute the equality constraint \eqref{equcon} into the objective function \eqref{objtran} and we get renewed problem formulation,
\vspace{-2mm}
\begin{align*}
\max  &  \sum_{i=1}^n (-G_{ik})\cdot[a^2+b^2-2ab\cdot\cos\theta_{ik}]\\
\rm{s.t.}  \  &V_{max}\leq a \leq V_{min}, \  \ V_{min}\leq b \leq V_{max}.
\vspace{-2mm}
\end{align*}
In the distribution grids, $G_{ik}$ is mostly negative, making $(-G_{ik})$ positive.
To find the maximum objective value in a non-convex problem, we employ a piecewise analysis of 
$2ab\cdot\cos\theta_{ik}$ according to $\theta$. 
We only show the reasonable upper and lower bounds of the angle difference in the following and the rest can be found in Appendix B.
\begin{itemize}
    \item $\theta \in [-\theta_{max},\theta_{max}]$, 
$\arccos(\frac{V_{max}+V_{min}}{2V_{max}}) \leq \theta_{max}\leq \pi$: \par
With voltage angle constraint, $\cos \theta_{max}\leq \cos\theta_{ik}\leq 1$. As $a$ and $b$ are positive, $\cos \theta_{ik}$ should be as small as possible to get a bigger value of the objective, which is $\cos \theta_{max}$. To maximize $a^2+b^2-2ab\cdot\cos\theta_{max}$, the voltage pattern should be at upper bounds ($a=b=V_{max}$) when $\cos \theta_{max}$ is a negative or small positive number. Until $\theta_{max}$ keeps decreasing to a specific number, the solution pattern changes to ``high-low voltage'' as in Theorem 1. For calculating this critical value, 
\vspace{-2mm}
\begin{align*}
V_{max}^2+V_{max}^2&-2\cdot V_{max} V_{max} \cos\theta_{max}\\
&\leq V_{max}^2+V_{min}^2-2 V_{max} V_{min} \cos\theta_{max},\\
\theta &\leq \arccos(\frac{V_{max}+V_{min}}{2V_{max}}).
\end{align*}
If we plug in the boundary values $V_{max}=1.05$ and $V_{min}=0.95$ below,
\vspace{-2mm}
\begin{align*}
 \vspace{-1mm}
 \theta \leq 0.3098 (rad). 
  \vspace{-4mm}
\end{align*}
 \vspace{-1mm}
Therefore, the solution for this angle range is
$a=b=V_{max}, \ \theta_{ik}=\theta_{max}$.
\item  $\theta \in [-\theta_{max},\theta_{max}]$, $ \theta_{max} < \arccos(\frac{V_{max}+V_{min}}{2V_{max}})$:  As we mention above, the solution pattern flips.
$a=V_{max}, b=V_{min},\  \theta_{ik}=\theta_{max}
$, when voltage angle is limited below the critical value $\arccos(\frac{V_{max}+V_{min}}{2V_{max}})$.
\end{itemize}
\end{proof}
 \vspace{-4mm}
When the voltage angle constraint narrows down, the variation of the results shows that the system acquires power generation prior to the difference of neighbor bus voltage angles. However, in the real systems, the angle difference between the two ends of a $100$km line is approximately $6$ degrees. Therefore, $\theta_{max}$ can be set to be a small value below $0.1047(rad)$ in the distribution grids, under which the second piece is meaningful. 

\subsubsection{HC due to Other Constraints} 
In addition to the voltage constraints (for $|V|$ and $\theta$), we obtain solution restricted to operational constraints. \\
\textbf{Solution Adjustment According to Thermal Limits }

 The thermal limit plays a significant role in two-way power flow.
Therefore, this constraint is represented by current flow limit that results from the heating effects of devices,
$I_{Lmin} \leq I_{ik} \leq I_{Lmax} = -I_{Lmin} = C$. 
We find the thermal limit boundary as inspired by the geometrical understanding and update the HC subject to both constraints.
\begin{theorem}
Consider the following thermal limit in addition to the constraints on top of Theorem $3$: $ I_{Lmin}\leq I_{ik} \leq I_{Lmax}, I_{Lmax}=-I_{Lmin}$. 
Let $I_{ik}^*$ be the branch current flow of previous solution pattern without thermal limit boundary. 
\begin{itemize}
\item If $I_{Lmin}\leq I_{ik}^* \leq I_{Lmax}$ is satisfied, the HC is achieved by the same voltage solution in Theorem $3$; 
\item Else, the HC is achieved by the solution pattern calculated from a three variable cubic equation in Appendix B.
\end{itemize}
\end{theorem}
\vspace{-2mm}
We will check if
the solution point in Theorem 3 encounters the violation  boundary or not. 
If not, the solution for HC stays the same. If yes, we will compute the new voltage solutions which is provided in Appendix B.\\
\textbf{Solution Adjustment According to Power Factor Limits}

PV (photovoltaics) generator produces DC power, meaning it injects power at unity power factor. 
Nonetheless, the local load and power supplied by the grid include reactive power, 
and PV inverter is 
also utilized to regulate the power factor by producing reactive power. Reactive power of all PV 
generators must lie within specific bounds, which in other words is to limit the power factor (Appendix B).
Similar to the thermal limit correction, we
substitute previous solution point to check if the inequality is true. If the reactive power is beyond the limit, we can switch bus type  
to fix reactive power at boundary value and compute the new voltage values and HC by Q-V sensitivity \cite{article}. 

Such a constraint is non-convex, but we can still find the theoretically optimal solution.
Other constraints can be analyzed following a similar way.
 \vspace{-2mm}
\section{Practical System Conditions}
In real system operation, our proposed model needs to be accommodated to different scenarios. We present two typical methodologies that integrate with other methods facing practical and operational difficulties.
 \vspace{-4mm}
\subsection{The Multi-phase system with Unbalanced Loads}
\vspace{-1mm}
The single-phase power flow model has a strong assumption that ignores the 
unbalance due to some unbalanced loads and untransposed lines. To provide a
tool that is more analytical and practical, we transfer multi-phase model
to sequence components and decouple them into separate subproblems. After 
that, our 
proposed model can be applied and transfer the results back to the origin \cite{abdel2005improved}.

Let $a$, $b$, and
$c$ denote three phases, the scalar variables of the single phase become vectors. 
$\mathbf{|V_i^{abc}|}=[|V_i^a|,|V_i^b|,|V_i^c|]^T$ and 
$\mathbf{\theta_i^{abc}}=[\theta_i^a,\theta_i^b,\theta_i^c]^T$ denote voltage magnitude and angle at bus $i$.
$\mathbf{P_i^{abc}}=[P_i^a,P_i^b,P_i^c]^T$ denotes the active power and so 
is $\mathbf{Q_i^{abc}}=[Q_i^a,Q_i^b,Q_i^c]^T$. The self and mutual line admittance are
considered, so a single element $Y_{ik}$ is replaced by a 
$3\times 3 $ matrix and then $Y_{Bus}^{abc}\in \mathbb{C^{\mathrm{3}\mathit{n}\times \mathrm{3}\mathit{n}}}$ ($Y_{ik}
\in \{Y^s_{ik},Y^z_{ik}\}$ if considering both series and shunt admittance). For single-phase or two-phase cases, simply set the quantities of the other phases to be zero and the 
model will degrade (e.g., $\mathbf{|V_i^{abc}|}=[|V_i^a|,0,0]$ for single phase.) 

We simplify the analysis of multi-phase system by transferring the model into the the sequence form.
In order to apply the unbalanced situation, we decouple the sequence model while using current injection from negative and zero sequences to represent the unbalance.

According to \cite{abdel2005improved}, the positive-sequence voltage magnitude is much larger under unbalanced condition and the power flow equations are similar to single-phase formulation. 
The equations allow us to implement the proposed method and obtain voltage pattern solutions for positive sequence.
Based on the decoupled models,
we can compute the negative and zero sequence voltages via nodal voltage equations.
Combining with results of positive sequence power flow, it's able to transfer back to the three-phase solution via \eqref{eq:trans} (we provide details in Appendix C).
\vspace{-4mm}
\subsection{The Distributed-Optimization for Complexity Reduction }
\vspace{-1mm}
For a large and complex system, the time on heuristically solving the HC is much longer as too many variables and constraints causing computational difficulties. 
In the following, we show that we can segment the tree structure and compute distributedly. 

If a system is divided into $m$ parts, the red segmentation point shown in 
Fig. 3 links 2 subsystems, it correlates to both parts geometrically. Mathematically in the problem model, the variables are divided into two sets. The first set includes local/private variables $l_{t}$ for subsystem $t$
and the second set includes coupling variables $c_{t}$ which bridge two more subsystems.
This turns the problem into the distributed formulation. 
\vspace{-2mm}
\begin{align}
\begin{split}
    \max &\sum_{t=1}^m  \ HC_{t}(l_{t},c_{t},c'_{t})   \\
\rm{s.t.} \  &c'_{t}=c_{t-1},~
V_{min}\leq |V| \leq V_{max},\\
&\theta_{min}\leq \theta \leq \theta_{max},
I_{Lmin} \leq I_{L} \leq I_{Lmax}.\label{segmodel}
\vspace{-4mm}
\end{split}
\end{align}
\vspace{-4mm}
\begin{figure}[H]
\centering
\vspace{-2mm}
\includegraphics[width=2.5in]{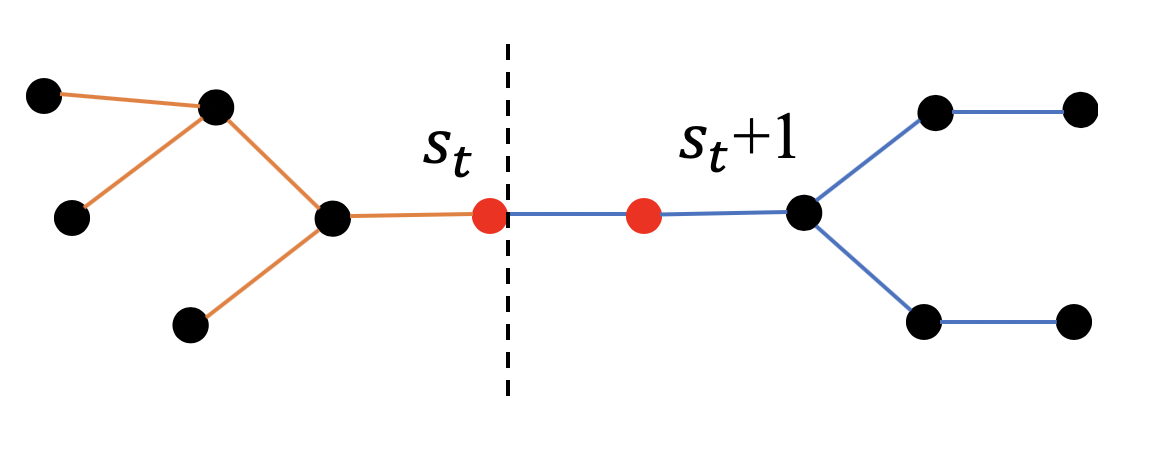}
\label{fig:seg}
\vspace{-3mm}
\caption{The system is segmented into $m$ parts for easier computation. $s_{t}$ is one segment bus linking 2 subsystems.
} 
\vspace{-2mm}
\end{figure}
\vspace{-3mm}
In the model,
$HC_{t}$ means the HC of subsystem $t$; $l_{t}$ represents the local variables - bus voltages that only appear in one subsystem. $c_{t}$ and $c'_{t}$ represent the complicating variables - bus voltages that appear in more than one subsystems. The objective function looses the coupling as it changes from the summation of power injections for all buses to the summation of HC for all subsystems. The coupling shifts to the constraints where the segmentation is true only if $c'_{t}=c_{t-1}$.
For each subsystem $t$, the objective is the same to analyze as \eqref{rec} (the proof is shown in the Appendix D.)
 Thus, it is capable of segmenting the model without disrupting the solution pattern in theorems and parallel computing.
\vspace{-3.5mm}
\section{Numerical Results}
 \vspace{-1mm}
For numerical verification, we use various distribution systems, such as modified IEEE $8$-bus and $123$-bus systems. 
We aim at $1$) validating the hosting capacity solution theory, $2$) evaluating the integrated method
on multi-phase unbalanced system， and
$3$) showing the effectiveness of parallel computation.

We use Matlab and its MatPower package for the demonstration \cite{matpower,matpower2}.
The optimization model is formed as non-linear programming, for which commonly used solvers include Fmincon, Knitro, and Ipopt. By comparing the results from one of the classical solvers with our theoretical solution, we find that their performances are unstable. In other words, they are easily stocked at local optimal especially when the system size is large and the initialization is inappropriate. 
\vspace{-5mm}
\subsection{Hosting Capacity}
\vspace{-1mm}
We use the IEEE $8$-bus system for illustrating HC according to various constraint conditions. Specifically, each bus is modeled with DG as a PV-type bus and we apply the proposed method to calculate HC.
\vspace{-2mm}
  \begin{figure}[H]
  \vspace{-2mm}
 \centering 
 \begin{subfigure}{0.47\textwidth}
 \centering
\includegraphics[width=3.3in]{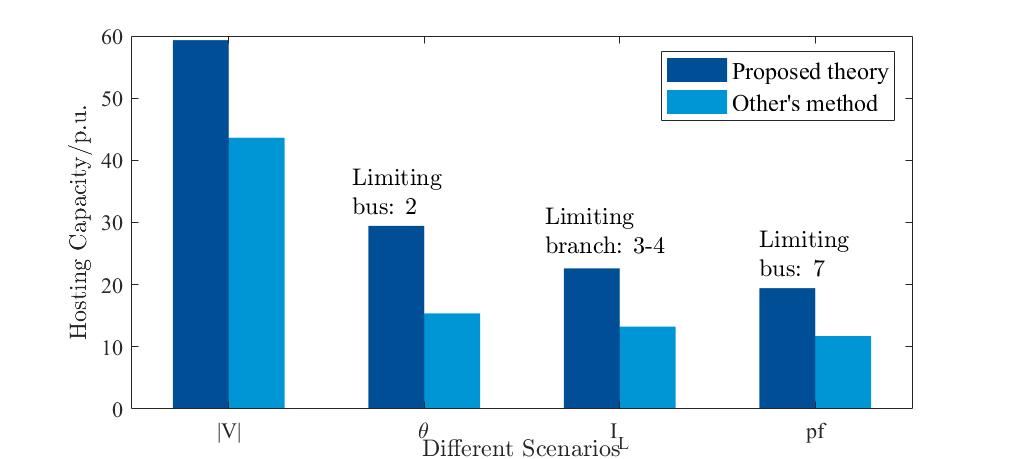}
 \label{fig:hc}
  \vspace{-2mm}
 \caption{The comparison of HC subject to different constraints.}
     \end{subfigure}
 \hfill
     \begin{subfigure}{0.45\textwidth}   
 \centering
 \includegraphics[width=3.2in]{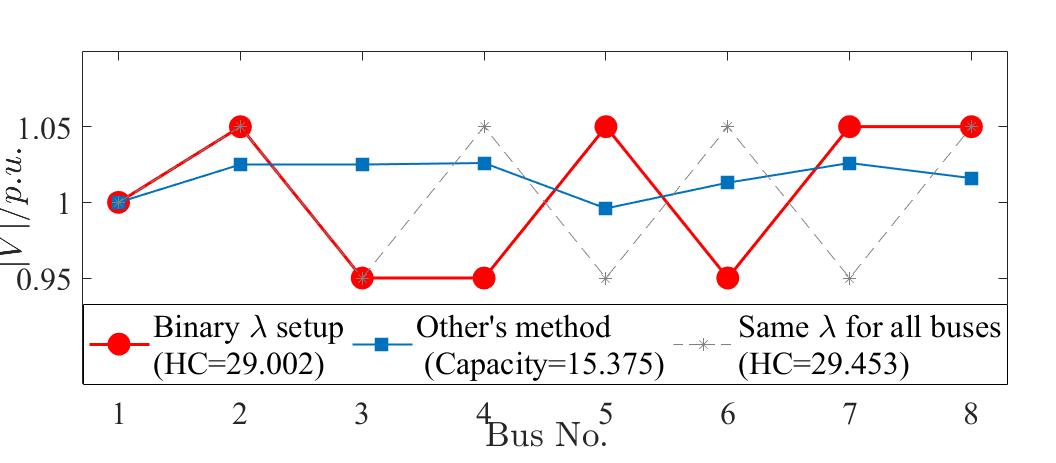}
  \vspace{-3mm}
 \caption{The voltage profiles with different $\lambda_{i}$ sets.}
     \end{subfigure}
     \vspace{-1mm}
         \caption{Numerical results for modified IEEE $8$-bus system.}
\vspace{-3mm}
 \end{figure}
  
Fig. 5a presents hosting capacity subject to different constraints and highlight the limiting bus node. The results of extensive simulation-based method are always smaller than ours. Considering different constraints, the HC from our theory under voltage limits is  $59.33 \ p.u.$, which is relatively large. Resulted from the thermal limit of the branch between node 3 and node 4, HC decreases to $22.63 \ p.u.$

Sometimes, only a few buses are interested or equipped with PV. For example, 
a binary $\lambda_i$ denotes if a bus is to be planned with PV or not. Here we 
set 
bus $2, 5$, and $7$ in the IEEE $8$-bus network, to have PV generators, e.g., $\lambda_{i}=1, i\in \{2,5,7\}$ and $\lambda_{i}=0, i \in \{1,3,4,6,8\}$. The objective turns into $\sum_{i=2,5,7} P_{i}$. 
 
As a result, Fig. 5b shows the corresponding voltage profile that compares to the original setup. As we can see that bus $2$, $5$, and $7$ are automatically assigned high voltages to generate as much as possible due to their 
generator bus type. For comparison, 
we also plot the voltage values by another method in MatPower. It simply runs power flow heuristically and stops at a local optimum or flatten region of the objective. The calculated HC value is $15.38 \ p.u.$, much less than $29.45 \ p.u.$, showing the importance of the closed-form solution proposed in this paper. To show the difference with respect to the first validation when all the $\lambda's$ are the same, we draw the voltage setup in gray in the figure as well. This shows how changing $\lambda_i$ to impact the way to achieve the hosting capacity. 
\newcommand{\tabincell}[2]{\begin{tabular}{@{}#1@{}}#2\end{tabular}}  
\vspace{-4mm}
\subsection{Multi-phase Unbalanced Condition}
\vspace{-5mm}
\begin{figure}[H]
 \centering
 \includegraphics[width=3.2in]{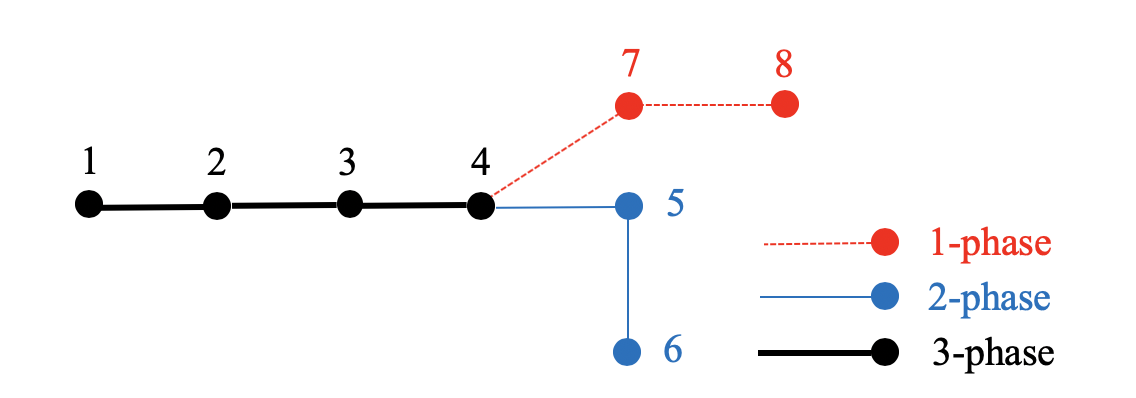}
 \label{fig:hl}
  \vspace{-3mm}
 \caption{The Multi-phase setup in IEEE 8-bus network.}
  \vspace{-2mm}
 \end{figure}
 \vspace{-2mm}
 \begin{table}[H]
\centering
\begin{tabular}{|c|c|c|}
\hline
{Scenario}&{Balance Condition} &{Method}\\
\hline
{All 1-phase}&Well-balanced & HC model \\
\hline
\multirow{2}{*}{All 2-phase }&Well-balanced& HC model \\
\cline{2-3}
&{Unbalanced load }&{Integrate with}  \\
&at bus $4$&{sequence power-flow}  \\
\hline
\multirow{2}{*}{All 3-phase}&Well-balanced&HC model \\
\cline{2-3}
&\multirow{2}{*}{Untransposed lines}& Integrate with  \\
&&  sequence line model  \\
\cline{2-3}
&{Unbalanced load }&{Integrate with}  \\
&at bus $4$&{sequence load current }  \\
\hline
\multirow{4}{*}{Multi-phase}&Well-balanced & HC model\\
\cline{2-3}
&\multirow{2}{*}{Untransposed lines}& Integrate with  \\
&&  sequence line model  \\
\cline{2-3}
&{Unbalanced load }&{Integrate with}  \\
&at bus $4$, $5$, $8$&{sequence load current}  \\
\hline
\end{tabular}
\caption{The evaluation of different scenarios for multi-phase setups. }
\end{table}
 \vspace{-4mm}
To evaluate the multi-phase modeling, we test different scenarios based on IEEE
$8$-bus unbalanced radial network in Fig. 6. The summary of solving methods is shown in Table \uppercase\expandafter{\romannumeral1}. If the system is well-balanced, 
meaning transposed lines and balanced loads, we simply 
apply the original model to compute hosting capacity. When the lines are 
unbalanced, meaning the admittance matrix in phase variables is unsymmetrical, 
we use sequence decoupled line model before applying the proposed method to 
positive sequence. Besides, the loads can also be unbalanced, where the 
load for each phase is individually specified instead of averaging from total 
demand. After the HC calculation via positive sequence power flow, the 
specified current injections represent such unbalance and hence the unbalanced phase voltages can be computed by nodal voltage equations.
 \vspace{-4mm}
\subsection{Parallel Computation}
\vspace{-1mm}
\begin{figure}[H]
    \centering
    \vspace{-1mm}
    \includegraphics[width=3.3in]{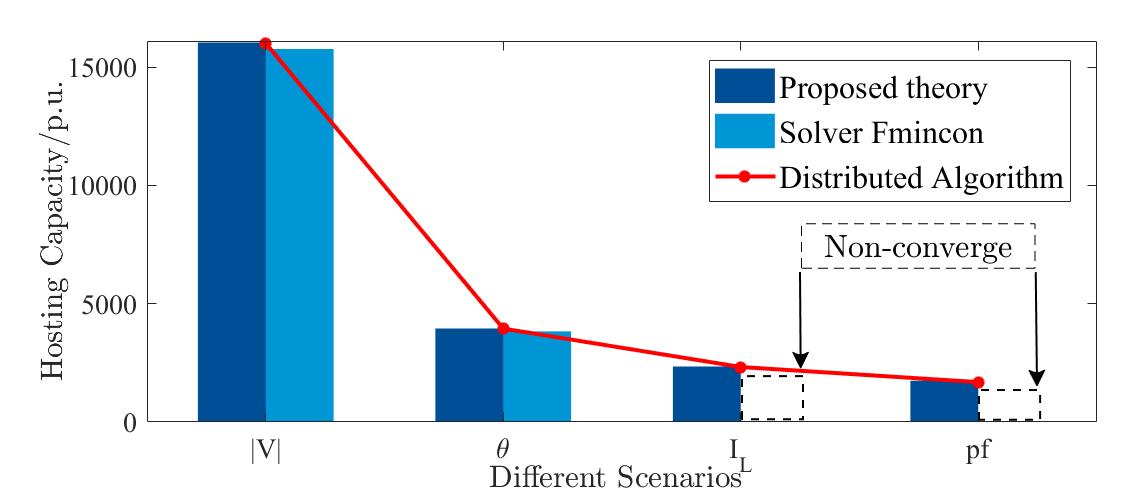}
    \caption{Hosting capacity of IEEE $123$-bus system subject to different constraints.}
\end{figure}
\vspace{-2mm}
Fig. 7 shows numerical results for IEEE $123$-bus system. Similar to the $8$-bus 
system, it shows HC subject to different constraints and compares the theoretical HC 
with the numerically found one by using NLP solver Fmincon for the same model, verifying 
our theorem. 
\begin{figure}[ht]
\vspace{-2mm}
\centering
\includegraphics[width=3.6in]{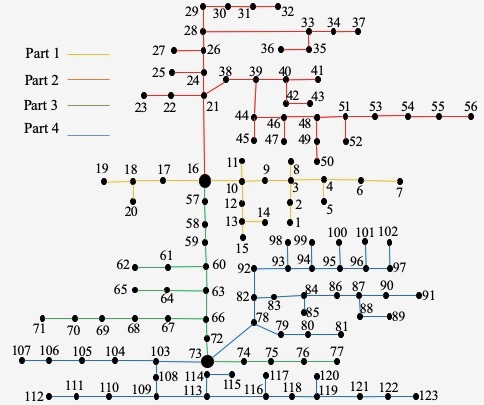}
\caption{The IEEE $123$-bus network structure shows the computational complexity. We implement the distributed optimization model for speed up. The different segmentations are marked with different colors.}
\vspace{-4mm}
\end{figure}

The benchmark time on heuristically solving HC is much longer and the solver is unstable because of the 
 much more complicated structure shown in Fig. 8. 
Too many variables and constraints cause
computational difficulties where the solver's results can be unstable and non-converge.
Applying the distributed optimization model \eqref{segmodel}, the segmentations in 
Fig. 8 are marked with different colors. Key nodes bus-$16$ and bus-$73$ are chosen 
to partition the model.
\vspace{-1mm}
\begin{figure}[H]
    \centering
    \includegraphics[width=3.3in]{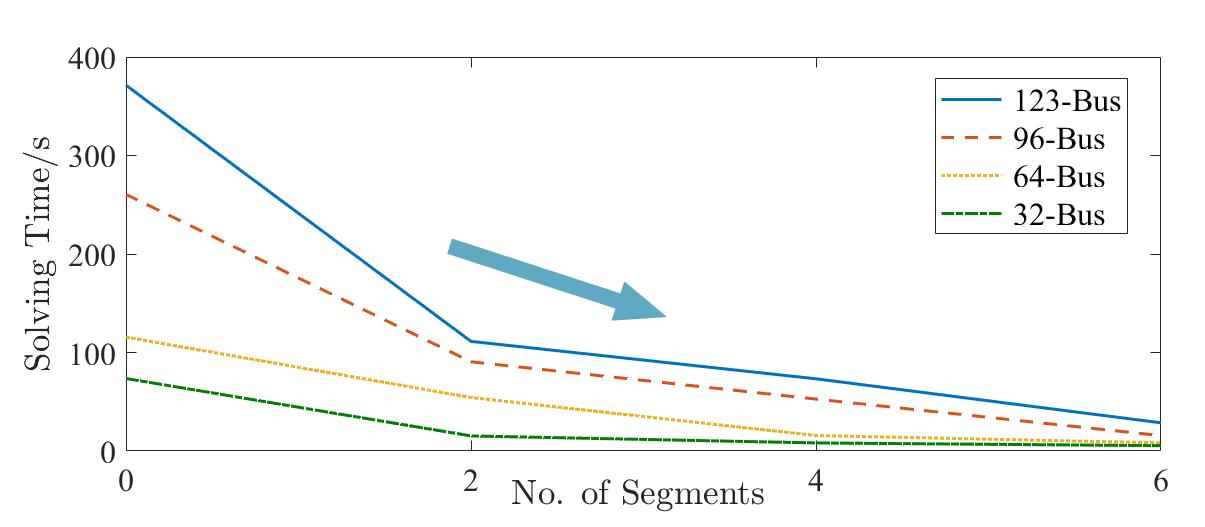}
 \vspace{-2mm}
    \caption{The comparison of solving time for different systems with different segmentations.}
    \vspace{-2mm}
\end{figure}
\vspace{-2mm}

The results with parallel computation are also shown in Fig. 7 as the
solid curve, which is almost the same as the bar plot, verifying the
distributed algorithm. The dramatic time decrease (Fig. 9) 
indicates we are able to reduce computational complexity and improve
efficiency as we calculate all the segments simultaneously.

\vspace{-4mm}
\section{Conclusion}
\vspace{-1mm}
In this paper, we propose a hosting capacity problem formulation for better solvability. A geometric intuition of solving HC motivates us to find the globally optimal solution even with non-convex constraints, with which we reveal the conditions of voltages to achieve the hosting capacity under various operational limits.   
Practical conditions for HC calculation, such as multi-phase unbalance and solving-time complexity, are analyzed via integrated methodologies.
A mathematical proof is provided along with numerical validations to support the  theoretical insights. 
Future work includes introducing more operation limits in a similar fashion and considering more complex modeling. 

\ifCLASSOPTIONcaptionsoff
  \newpage
\fi
\bibliographystyle{IEEEtran}

\appendices
\section{Implement on 3-Bus Toy Example}
For the simple setup to obtain geometric understanding, voltage angles and line reactances are set to zeros. In such a setup, the HC problem model simplifies to
\begin{align}
\begin{split}
\max &\sum_{i=1}^n  P_{i}^{ing}\\
\rm{s.t.} \
 &P_{i}=\sum_{k=1}^n {|V_{i}|}{|V_{k}|} G_{ik},\label{simplified2}\\
&V_{min}\leq |V_{i}| \leq V_{max}.
 \vspace{-3mm}
\end{split}
\end{align}
\section{Hosting Capacity Subject to Constraints}
\subsection{Proof of Theorem 1: Solution Adjustment According to Voltage Magnitude Constraints}
\begin{proof}
For the proof, it is sufficient to show that the proposed hosting value $HC^{*}$ is larger than arbitrary summation of power generation. Specifically, we need to show that 
\vspace{-2mm}
\begin{align}
\sum_{i=1}^n P_{i,\rm{solution}} \geq \sum_{i=1}^n {P_{i}}, \label{proof}
\vspace{-2mm}
\end{align}
where $\sum_{i=1}^n P_{i,\rm{solution}}$ 
represents the solution $HC^{*}$ according to the theorem. 
$\sum_{i=1}^n {P_{i}}$ represents the total power calculated by any setup of voltages within limits.
Expanding the left-hand side of \eqref{proof} with the voltages in the theorem 
\begin{subequations}  \label{eq:5}
\begin{align}
\sum_{i=1}^n P_{i,\rm{solution}}&= n\cdot|G_{ik}|\cdot {(1.05-0.95)}^2 \label{eq:5A}\\
&\geq \sum_{i=1}^n |G_{ik}|\cdot(V_{i}-V_{k})^2 \label{eq:5B}\\
& =\sum_{i=1}^n {P_{i}}, \label{eq:5C} 
\end{align}
\end{subequations}
where \eqref{eq:5A} is the direct plug-in, \eqref{eq:5B} is due to extreme numbers, and \eqref{eq:5C}  is due to power flow equations. 
\end{proof}
\subsection{Proof of Theorem 2: Solution Adjustment According to Angle Difference Constraint}

\begin{proof}
With the general complex impedance setup, there are two types of functionals, namely sinusoids polynomials in the polar coordinate. The multiplication of these two functionals in the power flow equation make the HC optimization  difficult to analyze for the global optimum. 
Therefore, we transform the problem into the rectangular coordinates to reduce the function types into one, namely, the polynomial functions. 
\begin{align}
 \begin{split}
 \max \sum_{i=1}^n P_{i} &  \\
\rm{s.t.} \  P_{i}=&-\sum_{k\in N_{i}} G_{ki}\cdot ( {V_{i,\textrm{real}}}^2+{V_{i,\textrm{imag}}}^2)\\
& +\sum_{k\in N_{i}}(G_{ki}V_{k,\textrm{real}}-B_{ki}V_{k,\textrm{imag}}) \cdot V_{i,\textrm{real}} \\
& +\sum_{k\in N_{i}}(V_{k,\textrm{real}}B_{ki}+V_{k,\textrm{imag}}G_{ki}) \cdot V_{i,\textrm{imag}},\\
V_{min}&\leq |V_{i}| \leq V_{max},\\
\theta_{min}&\leq \theta_{i} \leq \theta_{max},
\end{split}
\end{align}
where bus $k$ is the neighbor bus of bus $i$.
 Extend objective function for quadratic optimization \cite{Y2017geometric19}, we obtain 
\begin{align}
\begin{split}
\sum_{i=1}^NP_{i}=& \sum_{i=1}^{N}\sum_{k=1}^{N}(-G_{ik})\cdot[(V_{i,\textrm{real}}-V_{k,\textrm{real}})^2\\
&+(V_{i,\textrm{imag}}-V_{k,\textrm{imag}})^2] \\
= &\sum_{i=1}^{N}\sum_{k=1}^{N}(-G_{ik})\cdot {(V_{i}\cos\theta_{i}-V_{k}\cos\theta_{k})}^2\\
& +(-G_{ik})\cdot{(V_{i}\sin\theta_{i}-V_{k}\sin\theta_{k})}^2.
\end{split}
\end{align}
As the multiplication is the key component in the quadratic optimization, we can simplify the objective further by having 
${|V_{i}|}=a, ~{|V_{k}|}=b, ~\cos\theta_{i}=a_{1},~\sin\theta_{i}=a_{2}, ~\cos\theta_{k}=b_{1},~\sin\theta_{i}=b_{2}$.
Suppose bus $i$ is the even bus and bus $k$ is the odd bus.
Such notation reshape the objective into
\vspace{-2mm}
\begin{align}
\max & \sum_{i=1}^n(-G_{ik})\cdot\left[{(aa_{1}-bb_{1})}^2+{(aa_{2}-bb_{2})}^2\right] \label{proof2}\\
\rm{s.t.} & \  0.95\leq a \leq 1.05, \  \ 0.95\leq b \leq 1.05,\label{proof3}\\
& \ a_{1}^2+a_{2}^2=1, \ \  b_{1}^2+b_{2}^2=1.\label{proof4}
\vspace{-2mm}
\end{align}
Substitute the equality constraint \eqref{proof4} into the objective function \eqref{proof2} and we get renewed problem formulation,
\vspace{-2mm}
\begin{align}
\begin{split}
\max  &  \sum_{i=1}^n (-G_{ik})\cdot[a^2+b^2-2ab\cdot\cos\theta_{ik}]\\
\rm{s.t.}  \  &0.95\leq a \leq 1.05, \  \ 0.95\leq b \leq 1.05.
\vspace{-2mm}
\end{split}
\end{align}

In the distribution grids, $G_{ik}$ is mostly negative, making $-G_{ik}$ positive.
To find the maximum HC, we show in the following that we can find the global optimum in a non-convex problem by piecewise analysis of 
$2ab\cdot\cos\theta_{ik}$ according to $\theta$. 
\vspace{-1mm}
\begin{enumerate}
\item \textbf{When $\theta \in [0,\theta_{max}]$, for a fixed $\theta_{max}\in[\pi, 2\pi]$:} \par$\theta \in [0,\theta_{max}]$ means that 
$-1\leq \cos\theta_{ik} \leq 1.$
As $a$ and $b$ are positive voltage magnitudes, we obtain the maximal HC, when $\cos\theta_{ik}=-1$. Therefore, one element of the objective function becomes
$a^2+b^2-2ab\cdot\cos\theta_{ik}=(a+b)^2$. The maximum is achieved when
$ a=b=1.05$, $\theta_{ik}=\pi$,
meaning hosting capacity is under the condition that all the bus voltage magnitudes are the highest and the difference between neighbor bus voltage angles is $\pi$. Therefore, we need to change the voltage angle of $\pi$ within one branch interval, impossible in the real system. Meanwhile, the unconstrained voltage angle is unstable with observation from several attempts, so the next step is to narrow down the voltage angle.
\item \textbf{When $\theta \in [0,\theta_{max}]$, where 
$\arccos(\frac{V_{u}+V_{l}}{2V_{u}}) \leq \theta_{max}\leq \pi$:} \par
With voltage angle constraint, $\cos \theta_{max}\leq \cos\theta_{ik}\leq 1$. As $a$ and $b$ are positive, $\cos \theta_{ik}$ should be as small as possible to get a bigger value of the objective, which is $\cos \theta_{max}$. To maximize $a^2+b^2-2ab\cdot\cos\theta_{max}$, the voltage pattern is the same as above ($a=b=1.05$), when $\cos \theta_{max}$ is a negative or small positive number. Until $\theta_{max}$ keeps decreasing to a specific number, the solution pattern changes to ``high-low voltage'' as in Theorem 1. For calculating this critical value, 
\vspace{-2mm}
\begin{align}
\begin{split}
    V_u^2+V_u^2-2\cdot V_u V_u \cos\theta_{max}
&\leq V_u^2+V_l^2-2 V_u V_l \cos\theta_{max},\\
\theta &\leq \arccos(\frac{V_{u}+V_{l}}{2V_{u}}), 
\end{split}
\end{align}
where $V_u$ and $V_l$ are the upper and lower bounds.
we plug in the boundary values of $a$ and $b$ below,
\vspace{-2mm}
\begin{align}
 \theta \leq 0.3098. 
 \vspace{-4mm}
\end{align}
Therefore, the solution for this angle range is
$a=b=1.05,  \theta_{ik}=\theta_{max}$.
\item \textbf{When $\theta \in [0,\theta_{max}]$, where $ \theta_{max} < \arccos(\frac{V_{u}+V_{l}}{2V_{u}})$:} 

\par As we mention above, the solution is
$a=1.05, b=0.95, \theta_{ik}=\theta_{max}
$, when voltage angle is limited below the critical value $\arccos(\frac{V_{u}+V_{l}}{2V_{u}})$. The ``high-low voltage'' pattern appears again.
\end{enumerate}
\end{proof}
\subsection{Proof of Theorem 3: Solution Adjustment According to Thermal Limit}
In addition to the voltage constraints (for $|V|$ and $\theta$), the thermal limit is rather important in two-way power flow.
We use current flow limit to represent this constraint that results from the heating effects of devices,
$I_{Lmin} \leq I_{ik} \leq I_{Lmax} = -I_{Lmin} = C$. We do a comparison of the previous solution and the thermal limit to estimate if it has the impact on the final solution.
\begin{proof}
We will plug in the voltage solutions from Theorem $3$ to check if the new bounds are violated or not. If not, the solution for HC stays the same. If yes, we will find the new voltage solutions. 

Specifically, we plug in the voltage solutions into the current flow below
\vspace{-2mm}
\begin{align}
I_{ik}=|G_{ik}+jB_{ik}| \cdot (V_{i}-V_{k}).\label{eq:7}
\vspace{-2mm}
\end{align}

Now, compare the calculated current to the new boundary constraints. 
\begin{itemize}
\item If $-C\leq I_{ik}^* \leq C$, meaning the thermal limit has no impact on the solution;
\item Otherwise, 
we plug in \eqref{eq:7} into \eqref{eq:2e} to find the solution.
\begin{align}
\begin{split}
|G_{ik}+jB_{ik}|^2\cdot|(V_{i,\textrm{real}}-V_{k,\textrm{real}})+j(V_{i,\textrm{imag}}&-V_{k,\textrm{imag}})|^2\\
&\leq C^2,
\end{split}
\end{align}
\vspace{-1mm}
equivalent to
\vspace{-1mm}
\begin{align}
\begin{split}
[(V_{i,\textrm{real}}-V_{k,\textrm{real}})^2+j(V_{i,\textrm{imag}}&-V_{k,\textrm{imag}})^2]\\
&\leq C^2/(G_{ik}^2+B_{ik}^2).
\end{split}
\end{align}
The quadratic functions on the left-hand side are exactly one component in (6) to represent the objective $\sum_{i=1}^n P_{i}$. Therefore, the power flow impacted by thermal limit is 
\begin{align}
P_{ik}^*= (-G_{ik})\cdot[(V_{i,\textrm{real}}-V_{k,\textrm{real}})^2+(V_{i,\textrm{imag}}-V_{k,\textrm{imag}})^2].\label{eq:13}
\end{align}
We calculate the maximum value that this component can achieve via
\vspace{-2mm}
\begin{align}
{P_{ik}^*}_{max}= (-G_{ik})\cdot C^2/(G_{ik}^2+B_{ik}^2).
\end{align}
To obtain the solution pattern (values of $|V_{i}|$, $|V_{k}|$ and $\theta_{ik}$), we solve this three-variable-cubic equation.
\begin{align}
a^2+b^2-2ab\cdot\cos\theta_{ik}= C^2/(G_{ik}^2+B_{ik}^2), 
\label{eqn:three}
\end{align}
where $a$ and $b$ represent $|V_{i}|$ and $|V_{k}|$.

Due to the coupling of variables in the power flow equations, the changing solution pattern of one branch can have the impact on the rest buses to achieve maximum generation. The voltage pattern might transpose from ``high-low'' to ``low-high''.
According to the previous theorems and proofs, the ``high-low voltage'' pattern in one branch is obtained because of the quadratic function type in \eqref{eq:13}. 
Therefore, such transpose will not change the part of HC uncorrelated to this branch except the buses connected to it. In other words, only the quadratic functions containing the involved variables follow this change.
\end{itemize}
\end{proof}
\subsection{Proof of Solution Adjustment According to Power Factor Constraint}
Limiting Reactive power of PV generator is to limit the power factor \eqref{eq:2f}.
\begin{equation}
  \frac{P}{|S|} \geq \eta,
\end{equation}
Suppose we use $\eta=0.95$ as the operation limit,
\begin{align}
    \begin{split}
        \frac{P^2}{P^2+Q^2}&\geq \eta^2,\\
        Q^2&\leq \frac{1-\eta^2}{\eta^2}P^2,\\
        Q &\leq \frac{\sqrt{1-\eta^2}}{\sqrt{\eta^2}}P \ for \ Q>0,\\
         Q &\geq -\frac{\sqrt{1-\eta^2}}{\sqrt{\eta^2}}P \ for \ Q<0,
    \end{split}
\end{align}

\section{Extend to Three-phase Unbalanced Network}

We simplify the analysis of multi-phase system by transferring the model into sequence form via symmetrical component transformation matrix.
\begin{equation}
T={
\left[ \begin{array}{ccc}
1 & 1 & 1\\
1 & \alpha^2 & \alpha\\
1 & \alpha & \alpha^2
\end{array} 
\right ]} \ \text{where} \ \alpha=1\angle120^{\circ},
\end{equation}
a set of sequence variables is obtained, i.e.
\begin{equation}
    \mathbf{|V_i^{012}|}=[|V_i^0|,|V_i^1|,|V_i^2|]^T=T^{-1}\mathbf{|V_i^{abc}|}.\label{eq:trans}
\end{equation}
The sequence admittance matrix is
\begin{equation}
Y^{012}=T^{-1}Y^{abc}T. 
\end{equation}
In order to apply the unbalanced situation, we decouple the sequence model while using current injection from negative and zero sequence to represent the unbalance.
\begin{align}
    I^m=\big (\frac{S^{load}}{V^m}\big )^*+\sum_{k\in \mathcal{N}} \Delta I_k^m, \ m=0 \ \text{or} \ 2.
\end{align}
According to \cite{abdel2005improved}, the positive sequence voltage magnitude is much larger under unbalanced condition and the power flow is similar to single-phase formulation. 
Therefore, the positive sequence power flow equations are given as
\begin{subequations}  
\begin{align}
P_{i}^{1}=\sum_{k=1}^n {|V_{i}^{1}|}{|V_{k}^{1}|^T}(G_{ik}^{1}\cos\theta_{ik}^{1}+B_{ik}^{1}\sin\theta_{ik}^{1}), 
\\
Q_{i}^{1}=\sum_{k=1}^n {|V_{i}^{1}|}{|V_{k}^{1}|^T}(G_{ik}^{1}\cos\theta_{ik}^{1}-B_{ik}^{1}\sin\theta_{ik}^{1}).
\end{align}
\end{subequations}
The equations allow us to implement the proposed method and obtain voltage pattern solutions for positive sequence.
Based on the decoupled models,
we can compute the negative and zero sequence voltages via nodal voltage equations. 
\begin{align}
    Y^mV^m=I^m,\ m=0 \ \text{or} \ 2.
\end{align}
Combining with results of positive sequence power flow, it's able to transfer back to three phase solution via \eqref{eq:trans}.

\section{Enable Parallel Computation for Large Scale}

Following the distributed optimization model in  \uppercase\expandafter{\romannumeral4}. B, we derive the objective
\vspace{-1mm}
\begin{align}
\begin{split}
    \mathit{HC_{t}}=&\sum_{i=s_{t-1}}^{s_{t}}\sum_{k=s_{t-1}}^{s_{t}}(-G_{ik})\cdot[(V_{i,\textrm{real}}-V_{k,\textrm{real}})^2 \\
&+(V_{i,\textrm{imag}}-V_{k,\textrm{imag}})^2] \ (s_{0}=1)\nonumber\\
=& \sum_{i=s_{j-1}}^{s_{j}}\sum_{k=s_{j-1}}^{s_{j}}(-G_{ik})\cdot {(V_{i}\cos\theta_{i}-V_{k}\cos\theta_{k})}^2\\
& +(-G_{ik})\cdot{(V_{i}\sin\theta_{i}-V_{k}\sin\theta_{k})}^2,
\end{split}
\end{align}
which shows the same pattern as (6). Thus, The capability of parallel computation with the proposed model is proved.

\end{document}